\documentclass[conference]{IEEEtran}
\IEEEoverridecommandlockouts
\usepackage{cite}
\usepackage{amsmath,amssymb,amsfonts}
\usepackage{algorithm}
\usepackage{algpseudocode}
\usepackage{algpseudocodex}
\usepackage{enumitem}
\usepackage{graphicx}
\usepackage{textcomp}
\usepackage{svg}
\usepackage{xcolor}
\usepackage{amsthm} 
\usepackage{soul,color}
\usepackage{url}
\usepackage{diagbox}

\def\BibTeX{{\rm B\kern-.05em{\sc i\kern-.025em b}\kern-.08em
    T\kern-.1667em\lower.7ex\hbox{E}\kern-.125emX}}
\usepackage[font=footnotesize]{subfig}
\def\schemeAO{\textsc{AO}$^2$}
\def\schemeCCT{CCT}
\def\schemeHROA{HROA}

\def\energyImprovementCCT{66.99}
\def\energyImprovementHROA{76.03}

\def\lifeImprovementCCT{66.64}
\def\lifeImprovementHROA{76.85}



\begin{document}

\title{Renewables Power the Orbit? Achieving Sustainable Space Edge Computing via QoS-Aware Offloading}

\author{
\IEEEauthorblockN{
Xiaoyi Fan\IEEEauthorrefmark{1}\thanks{X. Fan and Y. C. Chou contributed equally to this research.}\thanks{Corresponding author: Ershun Du (duershun@tsinghua.edu.cn).},
Yi Ching Chou\IEEEauthorrefmark{2},
Hao Fang\IEEEauthorrefmark{2},
Long Chen\IEEEauthorrefmark{2},
Haoyuan Zhao\IEEEauthorrefmark{2},\\
Ershun Du\IEEEauthorrefmark{1},
Chongqing Kang\IEEEauthorrefmark{1},
Zhe Chen\IEEEauthorrefmark{3},
Jiangchuan Liu\IEEEauthorrefmark{2}
}

\IEEEauthorblockA{\IEEEauthorrefmark{1}Department of Electrical Engineering, Tsinghua University, China}

\IEEEauthorblockA{\IEEEauthorrefmark{2}School of Computing Science, Simon Fraser University, Canada}

\IEEEauthorblockA{\IEEEauthorrefmark{3}Institute of Space Internet, Fudan University, China}

Emails: fan-xy25@mails.tsinghua.edu.cn, ycchou@sfu.ca, fanghaof@sfu.ca, longchen.cs@ieee.org, \\
hza127@sfu.ca, duershun@tsinghua.edu.cn, cqkang@tsinghua.edu.cn, zhechen@fudan.edu.cn, jcliu@sfu.ca
}

\maketitle

\newtheorem{theorem}{Theorem}
\newtheorem{lemma}{Lemma}
\newtheorem{definition}{Definition}

\begin{abstract}

Low-Earth-Orbit (LEO) satellite constellations are becoming integral to 6G infrastructure, but increasing in-orbit computation accelerates battery degradation and raises sustainability concerns. Meanwhile, renewable-heavy regions worldwide experience persistent energy curtailment due to transmission bottlenecks, leaving substantial clean energy stranded near generation sites. We identify a satellite-grid co-design opportunity: adaptively offloading task-critical data from satellite to data centers co-located with renewable power plants. However, realizing this vision requires jointly considering intermittent and capacity-limited communication windows, as well as time-varying electricity budgets. In this paper, we propose SQSO, a Sustainable and QoS-aware Satellite Offloading framework that models per-interval task offloading as a constrained optimization over dynamic topology and electricity prices. Under this framework, we design \schemeAO{}, an adaptive offloading orchestration algorithm to solve the formulated optimization problem. Using Starlink-scale simulations and real-world electricity price traces, \schemeAO{} reduces energy consumption by up to \energyImprovementHROA{}\% and battery life consumption by up to \lifeImprovementHROA\% compared to state-of-the-art schemes, while also lowering task delay. This work highlights that sustainable scaling of LEO constellations requires co-design of space networking and renewable energy infrastructure, while our solution promotes renewable-aware task offloading and cross-domain collaboration for space-energy integration in the 6G era.
\end{abstract}

\section{Introduction}

Renewable energy, primarily wind and solar, has become central to global decarbonization, providing scalable electricity with near-zero operational emissions \cite{ipcc2022_ar6wg3_ch6, nrel2021_lca_electricity}. At the same time, networks and data centers already consume substantial and growing electricity, a trend expected to accelerate with 6G’s ultra-dense deployments and pervasive edge intelligence \cite{itu2024_ict_footprint, iea2023_datacentres_networks, itur2023_m2160}. Consequently, the sustainable scaling of future communication systems increasingly depends on the large-scale availability and efficient utilization of renewable energy \cite{ipcc2022_ar6wg3_ch6}.

Yet, despite rapid growth in renewable generation capacity, a substantial fraction of clean energy remains underutilized. In 2024, the California Independent System Operator curtailed 3.4 million megawatt-hour of utility-scale wind and solar generation, a 29\% year-over-year increase, largely because midday renewable output exceeded local demand and transmission capacity \cite{eia2025caiso, caiso2024prodcurt}. This imbalance reflects a spatial decoupling between the generation and demand of renewable energy. More specifically, renewable power plants are increasingly built in remote, resource-rich regions, while expanding substations and long-distance transmission to reach load centers remains slow and capital-intensive \cite{bird2014curtailment, yang2012windchina}. As a result, large volumes of clean energy remain stranded despite rising demand.

In the meantime, a similar sustainability issue arises in Low Earth Orbit (LEO) satellite systems. Growing sensing and edge-computing workloads increase onboard energy consumption and accelerate battery degradation \cite{xing2024cots, li2024phoenix}. Because battery replacement is rarely practical, operators often resort to launching replacement satellites, raising costs and shortening effective mission lifetimes \cite{li2024phoenix}. This coupling between workload, battery stress, and replacement cost makes energy efficiency highly important to sustainable operation \cite{michel2022starlink}.

\begin{figure}[t]
     \centering
     \includegraphics[width=\linewidth]{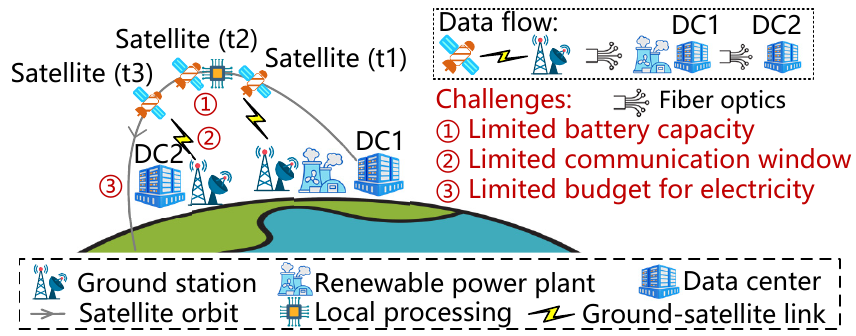}
     \caption{Task-critical data on-demand offloading from satellites to data centers near renewable power plants.}\vspace{-4mm}
     \label{fig:eo_mission_intro_overview}
\end{figure}

A natural opportunity emerges from the way modern infrastructure is built. Large data centers are commonly sited near energy hubs for capacity and cost reasons, and these hubs increasingly include renewable-heavy regions \cite{caict2025computingpower, morrison2023carbonexplorer}. Many such campuses are also co-located with (or built alongside) ground stations to support satellite connectivity. We therefore envision a framework in which satellites offload task-critical data \emph{on demand} to \emph{power-plant-adjacent} data centers for processing, trading scarce onboard energy for opportunistic ground compute powered by otherwise-curtailed renewables. As illustrated in Fig.~\ref{fig:eo_mission_intro_overview}, we trade off the energy consumption of in-orbit processing at \texttt{t2} against the potentially curtailed renewable energy near data center DC1. On-demand offloading of data to DC1 at \texttt{t1}, instead of processing in the orbit at \texttt{t2} and downlinking results at \texttt{t3}, reduces more in-orbit battery consumption while jointly improving the Quality of Service (QoS) of tasks and the satellite sustainability.

 \textbf{Challenge 1: Spatial-temporal renewable mismatch and on-board energy constraint.} Transmission limits cause persistent renewable curtailment and geographically concentrated \emph{stranded} energy, which, paired with limited on-board satellite energy, creates the core challenge of designing optimal offloading strategies to maximize sustainability and end-to-end QoS. \textbf{Challenge 2: Intermittent and capacity-constrained satellite-ground link.} Satellite mobility induces intermittent, capacity-limited contact windows, and contention at the satellite-ground boundary can prevent full offloading of task inputs, causing execution failures even when end-to-end routes exist \cite{vasisht2023umbra}. \textbf{Challenge 3: Volatile excess energy and tight budget constraint.} Available excess energy is highly unstable with sharp fluctuations in electricity prices and renewable output; naive offloading can exceed the fixed electricity budget and erase intended sustainability and QoS gains, requiring a balanced trade-off between performance, benefits and costs under satellite processing and budget constraints.

Prior work addresses pieces of this puzzle but leaves a key gap. On the \emph{terrestrial} side, geo-distributed load shifting and carbon-aware scheduling exploit electricity-price and carbon-intensity heterogeneity across data centers \cite{qureshi2009sigcomm, morrison2023carbonexplorer, hanafy2024gaia}. However, these systems do not explicitly model satellite contact windows, nor the combinatorial coupling between routing and offloading feasibility in satellite-ground fusion architectures. On the \emph{space} side, offloading and traffic-scheduling systems explicitly reason about intermittent connectivity and dynamic network topology \cite{vasisht2021l2d2, vasisht2023umbra, kassing2020hypatia, li2024short}. Yet, they typically treat the ground as a small set of fixed sinks, and do not explicitly incorporate the time-varying electricity prices and renewable energy availability of candidate processing sites into their core scheduling logic. Finally, \emph{emerging satellite edge-computing frameworks} optimize computation placement under power constraints \cite{xing2024cots, li2024phoenix}, but they do not explicitly incorporate renewable curtailment or electricity-price diversity at the terrestrial edge as key optimization considerations.

This paper proposes the Sustainable and QoS-Aware Satellite Offloading (SQSO) Framework, as well as the Adaptive Offloading Orchestration (\schemeAO{}) algorithm under this framework to maximize the gains in terms of sustainability and QoS. Both SQSO and \schemeAO{} leverage power-plant-adjacent data centers as flexible, opportunistic processing nodes to unlock otherwise curtailed renewable energy. 

At a high level, we build a time-sliced graph model of the dynamic satellite network and reduce per-interval decisions to a constrained assignment problem, enabling scalable re-optimization as topology and prices evolve. More spefically,  or each incoming task, our algorithm makes joint decisions on the intermediate ground processing site and the feasible space-to-ground transmission route, with design logic tightly aligned to the aforementioned three challenges: (i) it maximizes a composite utility function integrating end-to-end QoS and sustainability, to address \textbf{Challenge 1}; (ii) it strictly complies with the time-varying capacity limits of satellite-ground contact windows, to meet the hard constraint mentioned in \textbf{Challenge 2}; (iii) it keeps total electricity consumption within the predefined budget, to resolve the core trade-off raised in \textbf{Challenge 3}. We summarize the contributions as follows:

\begin{enumerate}
  \item We conduct a measurement-driven study of renewable curtailment and electricity price volatility at power-plant-adjacent data centers (often co-located with ground stations), revealing a satellite-grid co-design opportunity that improves task QoS and satellite battery lifetime.

  \item We introduce SQSO, a Sustainable and QoS-aware Satellite Offloading framework. We formulated an on-demand joint optimization problem of offloading orchestration under SQSO, with marginal performance gains in QoS and sustainability, satellite contact dynamics, and time-varying energy conditions jointly considered.

  \item We design the Adaptive Offloading Orchestration (\schemeAO{}) algorithm to solve the formulated problem. We implement the SQSO framework that integrates the \schemeAO{} algorithm based on real-world traces. Extensive evaluation shows that \schemeAO{} significantly improves battery life and task QoS compared to state-of-the-art methods.
\end{enumerate}

The remainder of this paper is structured as follows:
Section \ref{sec:motivation} briefly presents our motivations. Section \ref{sec:system_model_and_problem_formulation} introduces the system model and formulate the problem. In Section \ref{sec:algorithm}, we propose the algorithm, followed by the performance evaluation in Section \ref{sec:eval}. We brefily review the related work in Section \ref{sec:related_work}. Finally, we discuss the limitations of our work and conclude this paper in Section \ref{sec:conclusion}.

\section{Background and Motivation}
\label{sec:motivation}

\begin{figure}[t]
     \centering
     \includegraphics[width=\linewidth]{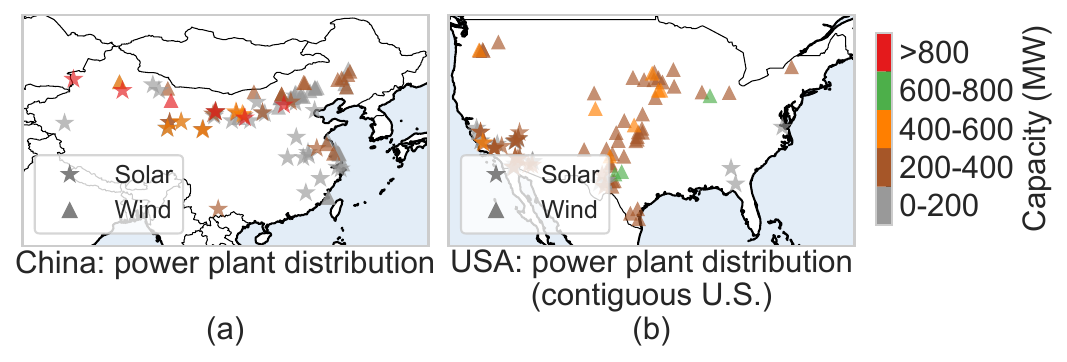}
     \caption{Spatial distribution and capacity of power plants in China and the contiguous United States.}
     \label{fig:power_plant}
     \vspace{-3mm}
\end{figure}

\subsection{Distribution Mismatch Between Renewables and Loads}
Wind and solar resources are inherently geographically skewed. Fig.~\ref{fig:power_plant} illustrates the distribution of power plants in China and the United States. In China, generation facilities are predominantly located in the northern and western regions. At the same time, in the United States, they are concentrated in central inland areas, where renewable resources are most abundant. In contrast, for both China and the United States, major cities and economic centers are largely located in coastal regions, far from these generation sites.

This mismatch extends beyond geography and manifests directly in system operations and utilization outcomes. Renewable-rich remote regions frequently experience higher curtailment rates than industrialized load centers, reflecting limited local demand and constrained transmission capacity for exporting power. Consequently, a significant fraction of renewable energy is generated far from where it can be readily consumed, leading to persistent \emph{excess energy} near production sites that is difficult to deliver to demand centers at scale.

\subsection{Integration Gap Between Renewables and Grids}

Fig.~\ref{fig:renew_capacity_and_waste}(a) illustrates the rapid global growth in the number of renewable power plants. The global installed capacity of solar and wind has experienced $2.64\times$ and $1.50\times$ increases from 2020 to 2024~\cite{OurWorldInData_RenewableCapacity}, reflecting sustained investment and policy support for renewable energy generation. However, this expansion in generation capacity has not been matched by commensurate growth in grid infrastructure or integration capability. As shown in Fig.~\ref{fig:renew_capacity_and_waste}(b), the average unutilized power ratio is 6.77\% (wind) and 9.16\% (solar) in Inner Mongolia, and 32.38\% (wind) and 35.05\% (solar) in Tibet~\cite{XNYXNYJ_Website}.

This gap between generation growth and delivery capacity stems primarily from structural transmission constraints. Integrating remote renewable resources requires slow and capital-intensive upgrades, including new transmission corridors, substations, and permitting processes \cite{bird2014curtailment}. As a result, the grid increasingly acts as a bottleneck, manifesting in rising curtailment across high-renewable markets. For example, in 2024, the California Independent System Operator curtailed 3.4 million megawatt-hour (MWh) of utility-scale wind and solar generation, a 29\% year-over-year increase, due to renewable output exceeding local demand and transmission capacity during high-production periods \cite{eia2025caiso}. Similar trends are observed in China, where official data reported by Reuters indicate a widening utilization gap in early 2025, with solar curtailment increasing year-over-year and some remote provinces experiencing substantially higher curtailment rates than eastern load centers \cite{reuters_china_renewable}.

Together, these patterns suggest that the marginal cost of fully delivering remote renewable energy is increasingly dominated by grid expansion, while a non-trivial fraction of clean energy remains stranded near its production sites.

\begin{figure}[t]
     \centering
     \includegraphics[width=\linewidth]{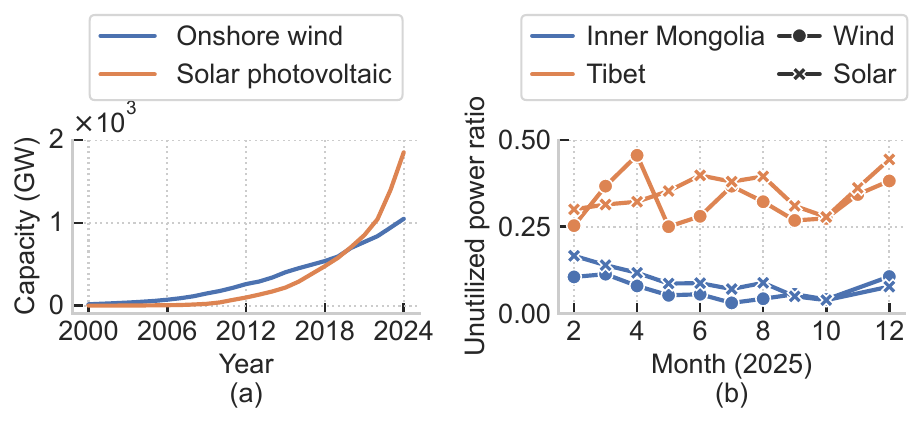}
     \caption{Growth in installed capacity of power plants, together with the unutilized power ratio over time.}
     \label{fig:renew_capacity_and_waste}
     \vspace{-4mm}
\end{figure}

\begin{figure}[t]
     \centering
     \includegraphics[width=\linewidth]{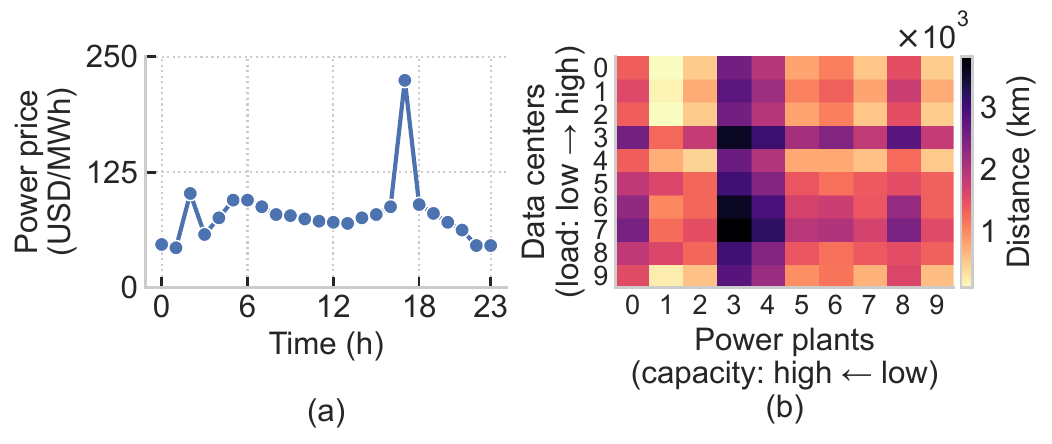}
     \caption{Electricity price over a day and the heatmap of distances between power plants and data centers.}
    \vspace{-4mm}
     \label{fig:power_price_and_distance}
\end{figure}

\begin{figure}[t]
     \centering
     \includegraphics[width=\linewidth]{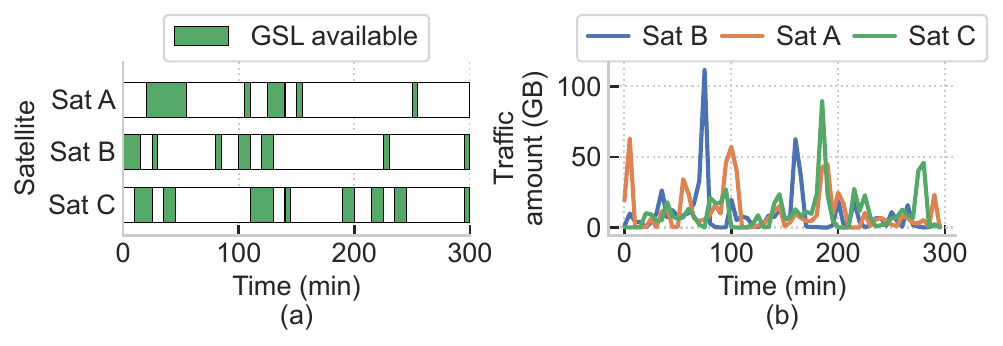}
     \caption{GSL availability windows of the satellite and the corresponding traffic volume over time.}
     \label{fig:satellite_window_and_traffic_amount}
     \vspace{-3mm}
\end{figure}

\subsection{Coupling Between Varying Prices and Local Compute}
Even when renewable energy is available, its \emph{effective} cost varies substantially over time. As shown in Fig.~\ref{fig:power_price_and_distance}(a), the average electricity price is only 79.50~USD/MWh, while peak prices reach 224.40~USD/MWh, nearly a 2.82$\times$ difference. This pronounced volatility reflects time-varying pricing mechanisms (e.g., time-of-use tariffs) combined with the diurnal and seasonal structure of renewable generation. As a result, the opportunity cost of consuming electricity fluctuates significantly across time, even within the same region.

At the same time, modern computing infrastructure is increasingly deployed in energy-advantaged locations. Fig.~\ref{fig:power_price_and_distance}(b) shows the spatial relationship between the ten most compute-intensive data centers and the ten largest power plants in China~\cite{Blackridge_TopDataCenters}. The results suggest strong spatial coupling: major compute clusters are frequently deployed near large generation sites, with an average separation of about 1583~km. The corresponding electricity transmission loss over this distance is negligible for long-haul ultra-high voltage transmission, while the distance from data centers to the regional transmission backbone is typically even smaller.

This co-location is particularly important for satellite offloading. Data centers near renewable plants can directly absorb time-local excess energy through computation, bypassing long-distance transmission and costly grid upgrades. Consequently, renewable-adjacent compute capacity provides a flexible, infrastructure-light way to convert otherwise underutilized renewable energy into useful work.

\subsection{Takeaways and Insights}

Our first insight is that \emph{offloading satellite tasks to renewable power plant-co-located data centers delivers both benefits in terms of sustainability and QoS}. This generation-side computation model directly consumes energy at its generation source, including otherwise curtailed renewable energy, which mitigates the long-standing spatial mismatch between renewable generation and load centers without relying on slow, capital-intensive transmission expansion. At the same time, reducing in-orbit computation cuts on-board energy consumption and mitigates battery degradation in LEO satellites by alleviating deep charge-discharge cycles. This design simultaneously addresses the core dual constraints of \textbf{Challenge 1}.

Our second insight reveals that \emph{intermittent satellite-ground link constraints and highly dynamic task arrivals are tightly coupled, and uncoordinated offloading decisions will fail to guarantee QoS or even cause task failures}. As shown in Fig.~\ref{fig:satellite_window_and_traffic_amount}(a), satellite mobility induces intermittent and capacity-limited contact windows, and contention at the satellite-ground boundary can prevent tasks from fully offloading required inputs, leading to execution failures even when end-to-end routes exist. Meanwhile, Fig.~\ref{fig:satellite_window_and_traffic_amount}(b) illustrates that satellite task demands arrive randomly with drastic fluctuations. Without coordinated scheduling, bursty traffic will rapidly exhaust the capacity of intermittent satellite-ground backhaul links, leading to severe degradation of task end-to-end QoS. This coupling creates a fundamental dilemma: the limited intermittent link capacity cannot match bursty dynamic task demands, and uncoordinated offloading will further exacerbate link contention and QoS degradation, which forms the constraint mentioned in \textbf{Challenge 2}.

Our third insight highlights that \emph{volatile energy conditions and rigid electricity budget constraints create an unavoidable trade-off between sustainability, QoS and operating costs, which cannot be resolved by naive uncoordinated offloading}. The available renewable excess energy and grid electricity prices fluctuate sharply over time with variable renewable output and grid load changes. Uncoordinated offloading decisions can rapidly exhaust the fixed electricity budget, completely erasing the intended sustainability benefits and QoS gains of ground-based processing. This core trade-off dilemma is exactly the key problem to be solved in \textbf{Challenge 3}.

Hence, we first introduce SQSO, a \emph{Sustainable and QoS-aware Satellite Offloading framework}. This framework is designed to accurately capture the impact of diverse offloading strategies on LEO satellite battery lifetime and end-to-end QoS, while quantifying the hard constraints of intermittent capacity-limited satellite-ground contact windows and fixed electricity budgets on offloading decision-making. Built on this framework, we propose the adaptive offloading orchestration algorithm, which is optimized to maximize the composite gains of system sustainability and end-to-end QoS, while systematically resolving all three core challenges.

\section{System Model and Problem Formulation}
\label{sec:system_model_and_problem_formulation}
\subsection{Preliminaries and Overview}
The terrestrial infrastructure is modeled as a graph $G_T = (V_T, E_T)$, where $V_T$ denotes the set of ground station nodes. Each ground station is assumed to be collocated with a data center, and the fiber-optic links among these ground stations and data centers are incorporated into the edge set $E_T$.

In practice, a data center may be supplied by more than one power plant, and electricity prices can vary across different power plants. Given the high complexity of real-world electricity market bidding mechanisms and due to page limitation, we abstract such complex electricity cost dynamics as a time-varying price $\texttt{price}_v(\tau)$ that is associated with each node $v \in V_T$ and time slot $\tau$.

We model the LEO satellite network as $G_S=(V_S, E_S(\tau))$, where $V_S$ denotes the set of satellites and $E_S(\tau)$ represents the set of ground-satellite links (GSLs), denoted as $E_{S}=\{l=(u,v) \mid u\in V_S, v\in V_T\}$. The bandwidth of GSL $l$ is denoted as $b_l(\tau)$. Considering the dynamics of the satellite network topology, we denote its orbital period as $\mathbb{T}$ and adopt a discrete-time model to partition $\mathbb{T}$ into discrete intervals. The topology is assumed stable within each interval $\tau \in \mathbb{T}$.

During each interval $\tau$, a set of tasks $R(\tau)$ is generated by satellites (e.g., remote sensing image processing tasks), each represented as a 4-tuple
\begin{equation}\label{eq:task_set}
\begin{split}  
R(\tau) =  \{r = (v_r^s,v_r^d,\xi_r,\zeta_r) \mid v_r^s \in V_S,& v_r^d \in V_E, \\
& \xi_r,\zeta_r \in \mathbb{R}^+\},
\end{split}
\end{equation}
where $v_r^s$ and $v_r^s$ denote the source and destination of the task, respectively; $\xi_r$ denotes the task execution time in cycles, and $\zeta_r$ denotes the data volume to be transmitted of the task, both being non-negative real numbers.

Tasks are generated by upper-layer applications and exhibit atomicity, meaning they cannot be partitioned further or processed across multiple data centers. We denote the budget for terrestrial task processing as $B(\tau)$ per interval.

\begin{figure}[t]
     \centering
     \includegraphics[width=0.8\linewidth]{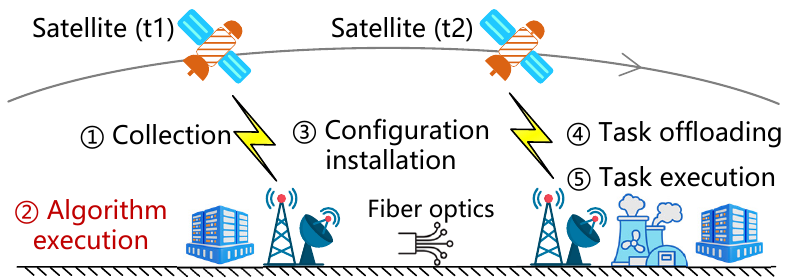}
     \caption{Overview of the Sustainable and
QoS-aware Satellite Offloading (SQSO) framework.}\vspace{-4mm}
     \label{fig:sqso_framework}
\end{figure}

As shown in Fig. \ref{fig:sqso_framework}, 
the Sustainable and QoS-aware Satellite Offloading (SQSO) framework works as below. At time $t_1$, the ground station collects essential status information, including pending tasks from the satellites. Subsequently, the proposed algorithm for solving the formulated optimization problem is executed on the ground. Upon obtaining the computational results, the configuration containing the task-offloading strategies is deployed to the target satellites. Thereafter, the configured satellites offload tasks to the designated ground stations, and these tasks are then transmitted to and executed within the data centers near the renewable power plants.

\subsection{Mobility-Aware Communication Model}
We define a binary indicator $x_r^l(\tau)$ for all $l \in E_{S}$ to denote whether the task $r$ choose the GSL $l$ to offload its data ($x_r^l(\tau) = 1$) or not ($x_r^l(\tau) = 0$).

Given the limited communication window between satellites and ground stations, we impose the bandwidth constraint
\begin{equation} \label{eq:bandwidth_constraint}
    \sum\nolimits_{r\in R(\tau)} (\zeta_r \times x_r^l(\tau)) \le  b_l(\tau) \times |\tau|,\ \forall l\in E_{S},
\end{equation}
where $|\tau|$ denotes the length of interval $\tau$. This constraint ensures that the total data offloaded from satellites to the ground within each interval does not exceed the transmission capacity of the GSLs.

\subsection{Capability-Based Cost Model}
We quantify the task processing cost using electricity fees. Specifically, we assume the power and processing capability measured in Hz allocated to task $r$ are provided by operators, denoted as $\texttt{power}_{r}^e$ and $\texttt{capability}_r^e$, respectively, where $v_r^e\in V_T$ represents the chosen power plant. For GPU tasks, we normalize the capability by multiplying the processing capability of each kernel by the number of kernels. It should be noted that the actual efficiency of GPU is related to parallel efficiency; the above normalization is a simplified estimation. In practical scenarios, the capability can be further adjusted according to the parallelization degree of specific tasks.
Based on this, the cost model for quantifying the expected payment of processing task $r$ is formulated as
\begin{equation}\label{eq:elec_fee}
    \theta_r(\tau) = \frac{\texttt{power}_{r}^e \times \xi_r \times \texttt{price}_{v_r^e}(\tau)}{\texttt{capability}_r^e}.
\end{equation}

Accordingly, the payment constraint stipulates that the total payment for scheduled tasks shall not exceed the budget, which is expressed as
\begin{equation} \label{eq:budget_constraint}
    \sum\nolimits_{r\in R(\tau)} \theta_r(\tau) \times y_r(\tau) \le B(\tau),
\end{equation}
where $y_r(\tau)$ is a binary indicator denoting whether task $r$ is scheduled in interval $\tau$. Specifically, $y_r(\tau)=1$ when task $r$ is scheduled, and $y_r(\tau)=0$ when it is not scheduled.

\subsection{QoS-Aware and Sustainability Utility Model}
The utility function consists of the quality-of-service (QoS) component and the sustainability component. Mathematically, it is expressed as
\begin{equation}
U(\tau)=U_{\texttt{QoS}}(\tau) + U_{\texttt{sustainability}}(\tau).
\end{equation}

\textbf{QoS-aware utility.} In this paper, $U_{\texttt{QoS}}(\tau)$ is defined as the ratio of the number of scheduled tasks to the total end-to-end delay from the power plant to the destination. We have 
\begin{equation}\label{eq:qos_utility}
    U_{\texttt{QoS}}(\tau)=\frac{|y_r(\tau)|}{\sum\nolimits_{r\in R(\tau)}x_r^l(\tau)\times \texttt{Delay}_{v_r^e\rightarrow v_r^d}},
\end{equation}
where $v_r^e$ is determined by $x_r^l(\tau)$ (i.e., $x_r^l(\tau)=1$ implies $l=(v_r^s, v_r^e)$), and $\texttt{Delay}_{v_r^e\rightarrow v_r^d}$ is estimated by dividing the distance between $v_r^e$ and $v_r^d$ by the speed of light. The total end-to-end delay is placed in the denominator to constrain the selected power plant from being excessively distant from the destination of the task.

Notably, for the numerator in (\ref{eq:qos_utility}), each satellite assigns task priorities based on application requirements, determining data offloading sequences. However, our $U_{\texttt{QoS}}(\tau)$ metric may introduce fairness concerns: for instance, if two satellites can establish GSLs to the same ground station, the one with lower-cost tasks (i.e., the expected electricity fee in Formula (\ref{eq:elec_fee})) is more likely to be scheduled for more tasks.

Nevertheless, our framework still has the potential to outperform traditional schemes, as conventional approaches require all data to wait on the satellite until it flies over a destination-proximal ground station. Additionally, fairness can be easily addressed by incorporating a fairness penalty into $U_{\texttt{QoS}}(\tau)$. Due to page limitations, we do not elaborate on fairness design herein and instead focus on the on-demand offloading orchestration framework.

\textbf{Sustainability utility.} For $U_{\texttt{sustainability}}(\tau)$, we focus on the impact of task processing on satellite battery lifespan.

Each LEO satellite harvests fixed power during sunlight periods and no power during eclipse periods \cite{yang_towards_2016}. Over interval $\tau$, the harvested energy of satellite $v$ is denoted as $\mathbb{E}_v^+(\tau)$. The initial battery level at the start of $\tau$ is $\mathbb{E}_v^0(\tau)$, and the baseline energy consumption for routine operations is $\mathbb{E}_v^-(\tau)$. The power for task processing is denoted as $\texttt{power}_v$.

Hence, if the task $r$ located in satellite $v$ is scheduled, the energy consumption for task processing can be saved and the remaining battery energy at the end of $\tau$ grows. We have
\begin{equation}\label{eq:remaining_energy}
\begin{split}
    \mathbb{E}_v^{\mathrm{rem}}(\tau) &= \max \{\mathbb{E}_v^0(\tau) + \mathbb{E}_v^+(\tau) - \mathbb{E}_v^-(\tau)\\
    &-\sum\nolimits_{r\in R(\tau), v_r^s=v} (1-y_r(\tau)) \mathbb{E}_v^{\texttt{process}}(\tau),0\},
\end{split}
\end{equation}
where $\mathbb{E}_v$ is a time-independent variable representing the battery capacity of satellite $v$ and $\mathbb{E}_v^{\texttt{process}}(\tau)$ can be obtained based on its processing capability $\texttt{capability}_v$ as below
\begin{equation}
   \mathbb{E}_v^{\texttt{process}}(\tau) = \frac{\texttt{power}_v \times \xi_r}{\texttt{capability}_v}.
\end{equation}
 
We focus on lithium-ion batteries, where cycle life is dependent on Depth-of-Discharge (DoD) \cite{yang_towards_2016}. For other battery chemistries (e.g., Ni-MH), similar models can be adopted with corresponding DoD-wear functions. The DoD of satellite $v$ at the start and end of interval $\tau$ are defined respectively as
\begin{equation}
    D_v^{\mathrm{begin}}(\tau) = \frac{\mathbb{E}_v - \mathbb{E}_v^0(\tau)}{\mathbb{E}_v},
\end{equation}
\begin{equation}\label{eq:dod_end}
    D_v^{\mathrm{end}}(\tau) = \frac{\mathbb{E}_v - \mathbb{E}_v^{\mathrm{rem}}(\tau)}{\mathbb{E}_v}.
\end{equation}

For lithium-ion batteries, we extend the lifespan consumption model based on \cite{yang_towards_2016} as follows
\begin{equation}\label{eq:life_consume}
L_v(\tau) = \int_{D_v^{\rm begin}(\tau)}^{D_v^{\rm end}(\tau)} 10^{0.8\,(D-1)}\bigl(1 + 0.8\,D\,\ln10\bigr)\,\mathrm{d}D.
\end{equation}

The integrand \(10^{0.8\,(D-1)}(1 + 0.8D\ln10)\) captures the nonlinear wear characteristic of lithium-ion chemistry. For alternative battery types, this kernel should be replaced with the corresponding DoD-dependent degradation function.

To minimize the impact of task processing on satellite battery lifespan, $U_{\texttt{sustainability}}(\tau)$ is quantified as the negative value of the sum of lifespan consumption across relevant satellites. It is formulated as:
\begin{equation}\label{eq:utility_sus}
   U_{\texttt{sustainability}}(\tau) = \sum\nolimits_{v\in V_S} -L_v(\tau).
\end{equation}

We formulate the problem as an on-demand joint optimization problem of offloading orchestration. Specifically, the problem aims to adaptively determine the optimal data center near a power plant for task processing, with joint consideration of the limited communication windows caused by satellite mobility and the constrained budget. It simultaneously improves battery lifespan preservation and QoS in each interval $\tau$. The decision variables include the target power plant $v_r^e$ ($\forall r\in R(\tau)$), which is determined by the binary indicator $x_r^l(\tau)$ where $x_r^l(\tau)=1$ implies $l=(v_r^s,v_r^e)$, and the scheduled tasks $y_r(\tau)$ ($\forall r\in R(\tau)$). The optimization problem is defined as
\begin{equation}
\label{eq:objective} 
\textbf{maximize}\quad U(\tau)
\end{equation}
\begin{equation*}
 s.t. \quad (\ref{eq:task_set})-(\ref{eq:utility_sus}),
 \end{equation*}
\begin{equation} \label{cons1}
x_r^l(\tau)\in\{0,1\}, y_r(\tau)\in\{0,1\},\forall r \in R(\tau), l\in E_S(\tau),
\end{equation}
\begin{equation}\label{cons2}
   x_r^l(\tau)\le y_r(\tau),\forall r \in R(\tau), l\in E_S(\tau),
\end{equation}
\begin{equation}\label{cons3}
    \sum\nolimits_{l\in E_T} x_r^l \le 1, \forall r\in R(\tau),
\end{equation}
where \eqref{eq:objective} specifies the optimization goal, \eqref{cons1} defines two binary variables and \eqref{cons2} indicates that $x_r^l(\tau)$ can be set as 1 only when the task $r$ is scheduled, i.e., $y_r(\tau)=1$. The constraint \eqref{cons3} guarantees the atomicity of each task.

\begin{theorem} \label{thm1}
The formulated problem is NP-hard
\end{theorem}
\begin{proof}
We consider a special case of the problem \eqref{eq:objective} where the battery lifespan utility $U_{\texttt{sustainability}}(\tau)$ is fixed and only the QoS utility $U_{\texttt{QoS}}(\tau)$ is maximized under the bandwidth constraint \eqref{eq:bandwidth_constraint}, budget constraint \eqref{eq:budget_constraint}, and task atomicity constraint \eqref{cons3}.
In this case, each task corresponds to an atomic item with a transmission cost, an execution cost, and a utility gain, while the constraints correspond to multiple knapsack capacities.
This special case reduces to the \emph{multi-constraint 0-1 knapsack problem}, which is known to be NP-hard.
The reduction can be accomplished in polynomial time, which implies that the formulated problem is NP-hard.
\end{proof}

\section{Adaptive Offloading Orchestration: Algorithm and Analysis}\label{sec:algorithm}
Given the NP-hardness of the formulated optimization problem, we propose a greedy algorithm to solve it. This algorithm comprises three core steps: first, it relaxes the communication window and budget constraints to determine the set of candidate satellites; second, it schedules tasks generated by these candidate satellites with joint consideration of communication and budget constraints; third, it offloads the scheduled task data to the power plant-collocated data center, settle the electricity fees, and execute the tasks. Unscheduled tasks on the satellite will either be rescheduled in the next interval via our algorithm or transmitted to the destination when the satellite passes over the ground station proximal to the destination.

\subsection{Candidate Satellite Set Construction}
Considering the ever-expanding scale of LEO satellite networks, we construct a candidate satellite set to reduce the problem scale. This set is an ordered collection containing all satellites that generate tasks. Formally, it is defined as
\begin{equation}\label{eq:overline_v_s}
    \overline{V}_S(\tau)=\{v \mid v\in V_S, \exists r\in R(\tau),v=v_r \}.
\end{equation}

To align with the problem's optimization objective, we sort $\overline{V}_S(\tau)$ such that the proposed algorithm can take this ordered set as input to maximize utility. Prior to presenting the sorting criteria, we first define several key concepts, all derived by relaxing the communication window and budget constraints.

For a satellite $v\in \overline{V}_S(\tau)$, the set of tasks it generates is denoted as
\begin{equation}\label{eq:r_v_tau}
    R_v(\tau) = \{r \mid r\in R(\tau)\}, \forall v\in \overline{V}_S(\tau).
\end{equation}

We define the normalized expected electricity fee for executing tasks in $R_v(\tau)$ as
\begin{equation}
    \overline{\theta}_v(\tau) = \frac{\sum_{r\in R_v(\tau)}\theta_r(\tau)}{|R_v(\tau)|\max_{r\in R_v(\tau)}\theta_r(\tau)}.
\end{equation}

Similarly, the normalized data volume of all tasks in $R_v(\tau)$ is defined as
\begin{equation}
    \overline{\zeta}_{v}(\tau) =\frac{\sum_{r\in R_v(\tau)}\zeta_r}{|R_v(\tau)|\max_{r\in R_v(\tau)}\zeta_r}.
\end{equation}

Finally, we need to quantify the impact of offloading task data to the data center on utility. Specifically, a satellite may establish GSLs with multiple ground stations; we adopt a greedy strategy to select the ground station that minimizes the end-to-end delay. This optimal ground station is 
\begin{equation}
    v_r^{e,*}=\arg\min\nolimits_{v_r^e\in V_T} \texttt{Delay}_{v_r^e\rightarrow v_r^d}
\end{equation}

Accordingly, the impact on QoS utility can be quantified as
\begin{equation}
\Delta U_{\texttt{QoS}}^v(\tau) = \frac{|R_v(\tau)|}{\sum_{r\in R_v(\tau)}\texttt{Delay}_{v_r^{e,*}\rightarrow v_r^d} },
\end{equation}
where $|R_v(\tau)|$ is placed in the numerator to estimate the QoS utility improvement if all tasks on satellite $v$ are scheduled.

For the impact on sustainability utility, we modify Formula \eqref{eq:remaining_energy} to calculate the remaining battery energy when tasks are processed locally (without offloading to the data center), which can be denoted as
\begin{equation}\label{eq:remaining_energy_new}
\begin{split}
        \mathbb{E}_v^{\mathrm{rem}, \mathrm{local}}(\tau) = & \min\{\mathbb{E}_v^0(\tau) + \mathbb{E}_v^+(\tau) - \mathbb{E}_v^-(\tau)\\
        &-\sum\nolimits_{r\in R(\tau), v_r^s=v}  \mathbb{E}_v^{\texttt{process}}(\tau),0\},
        \end{split}
\end{equation}
where the superscript $\mathrm{local}$ indicates local task processing on the satellite. 

By substituting $\mathbb{E}_v^{\mathrm{rem}, \mathrm{local}}(\tau)$ for $\mathbb{E}_v^{\mathrm{rem}}(\tau)$ in Formula \eqref{eq:dod_end}, we can derive the expected lifespan consumption via Formula \eqref{eq:life_consume}, denoted as $L_v^{\mathrm{local}}(\tau)$. Thus, the impact on sustainability utility is expressed as
\begin{equation}
    \hspace{-1mm}\Delta U_{\texttt{sustainability}}^v(\tau) = L_v^{\mathrm{local}}(\tau)-L_v(\tau), \forall v\in \overline{V}_S(\tau).
\end{equation}

Based on the above analysis, we sort $\overline{V}_S(\tau)$ in non-increasing order of the marginal utility gain $m_v(\tau)$. This gain is defined as the ratio of the sum of QoS utility and sustainability utility improvements to the sum of normalized expected electricity fees and normalized data volume. Formally
\begin{equation}\label{eq:sorting_criteria}
    m_v(\tau) = \frac{\Delta U_{\texttt{QoS}}^v(\tau)+ \Delta U_{\texttt{sustainability}}^v(\tau)}{\overline{\theta}_v(\tau)+\overline{\zeta}_{v}(\tau)}.
\end{equation}

Notably, the normalized expected electricity fees and normalized data volume are placed in the denominator to measure the utility improvement per unit of budget and GSL bandwidth consumption.

\subsection{Algorithm Design}
The proposed Adaptive Offloading Orchestration (\schemeAO{}) algorithm takes the topologies of the satellite network and terrestrial infrastructure, along with other necessary information (e.g., electricity prices $\{\texttt{price}_v|v\in V_T\}$), as inputs.

In the candidate satellite set construction stage (lines \ref{algline:stage1_begin}–\ref{algline:stage1_end}), we first initialize $\overline{V}_S(\tau)$ and $R_v(\tau)$ in accordance with \eqref{eq:overline_v_s} and \eqref{eq:r_v_tau}, respectively. Subsequently, $\overline{V}_S(\tau)$ is sorted in non-increasing order of $m_v(\tau)$, where $m_v(\tau)$ is calculated via \eqref{eq:sorting_criteria}.

In the orchestration stage (lines \ref{algline:stage2_begin}–\ref{algline:stage2_end}), the algorithm selects the first element of the sorted $\overline{V}_S(\tau)$ (i.e., the satellite $v$ with the maximum $m_v(\tau)$) and focuses on the tasks generated by $v$. An empty set $V_{\texttt{gs}}$ is initialized; a candidate power plant $v_r^e$ is added to $V_{\texttt{gs}}$ if it satisfies both constraints \eqref{eq:bandwidth_constraint} and \eqref{eq:budget_constraint}.

Specifically, after evaluating all GSLs, the algorithm checks if $V_{\texttt{gs}}$ is non-empty. If so, the target power plant is selected as the one with the minimal expected electricity fee (line \ref{algline:core_function}). Meanwhile, the corresponding indicators $x_r^l(\tau)$ (for $l=(v_r^s,v_r^e)$) and $y_r(\tau)$ are updated to 1. It should be noted that $x_r^l(\tau)$ is set to 0 in line \ref{algline:zero_setting} to ensure that each task is processed by only one data center. If $V_{\texttt{gs}}$ is empty, both indicators are set to 0 (line \ref{algline:corner_case}).

Finally, the algorithm removes $v$ from $\overline{V}_S(\tau)$ (line \ref{algline:removal_and_update}) and repeats the aforementioned process until $\overline{V}_S(\tau)$ is empty, indicating that all candidate satellites generating tasks have been considered.

Notably, there is no need to explicitly update the remaining transmittable data of GSLs and the remaining budget after scheduling a task. This is because both indicators are configured in lines \ref{algline:updating_begin}–\ref{algline:updating_end}, and both constraints \eqref{eq:bandwidth_constraint} and \eqref{eq:budget_constraint} already take these two indicators into account. An overview of the proposed algorithm is presented in Fig. \ref{fig:overview_algo}.

Furthermore, although task priorities are specified by upper-layer applications, the algorithm still iterates over all tasks in $R_v(\tau)$ even if a task cannot be scheduled (i.e., $V_{\texttt{gs}}$ is empty after constraint checks). The rationale is to further optimize the objective by verifying whether tasks associated with satellites having smaller $m_v(\tau)$ (indicating a lower ratio of total utility improvement to the sum of normalized expected electricity fees and normalized data volume) can be scheduled.

\begin{algorithm}[t]
\textbf{Input:} \parbox[t]{0.9\linewidth}{$G_S, G_T, R(\tau), \{\texttt{price}_v|v\in V_T\},\{\mathbb{E}_v|v\in V_S\}\\
\{\texttt{power}_r^e,\texttt{capability}_r^e|r\in R(\tau), v_r^e\in V_T\},\\
\{\texttt{power}_v,\texttt{capability}_v|v\in V_S\}$} \\
\textbf{Output:} $\{x_r^l,y_r|\forall r\in R(\tau), l\in E_S(\tau)\}$
% }
\caption{Adaptive Offloading Orchestration (\schemeAO{})} \label{algo:adaptive_offloading}
\begin{algorithmic}[1]
\State{$\text{Initialize}~\overline{V}_S(\tau)~\text{and}~R_v(\tau)~\text{by}~\eqref{eq:overline_v_s}~\text{and}~\eqref{eq:r_v_tau}$}\label{algline:stage1_begin}
\State {Sort $\overline{V}_S(\tau)$ in the non-increasing order of $m_v(\tau)$}\label{algline:stage1_end}
\While{$\overline{V}_S(\tau)\neq \emptyset$} \label{algline:stage2_begin}
\State{$v\leftarrow~\text{the first element of}~\overline{V}_S(\tau)$}
\For{$r\in R_v(\tau)$}
\State {$V_{\texttt{gs}}=\emptyset$}
\For{$l=(v_r^s,v_r^e)\in E_S(\tau)$}
\If{Constraints \eqref{eq:bandwidth_constraint} and \eqref{eq:budget_constraint} are satisfied}
\State{$V_{\texttt{gs}}=V_{\texttt{gs}}\cup \{v_r^e\}$}
\EndIf
\EndFor
\If{$V_{\texttt{gs}}\neq \emptyset$} \label{algline:updating_begin}
\State {$v_r^e = \arg\min_{v_r^e\in V_{\texttt{gs}}} \theta_r(\tau)$}\label{algline:core_function}
\State {$x_r^l(\tau)=1~\text{where}~l=(v_r^s,v_r^e),y_r(\tau)=1$}
\State{$x_r^l(\tau)=0,\forall l\in E_S(\tau)\setminus\{(v_r^s,v_r^e)\}$}\label{algline:zero_setting}
\Else
\State {$x_r^l(\tau)=0,\forall l\in E_S(\tau), y_r(\tau)=0$}\label{algline:corner_case}
\EndIf  \label{algline:updating_end}
\EndFor
\State{$\overline{V}_S(\tau)=\overline{V}_S(\tau)\setminus \{v\}$}\label{algline:removal_and_update}
\EndWhile \label{algline:stage2_end}
\end{algorithmic}
\end{algorithm} 

\begin{figure}[t]
     \centering
     \includegraphics[width=.5\linewidth]{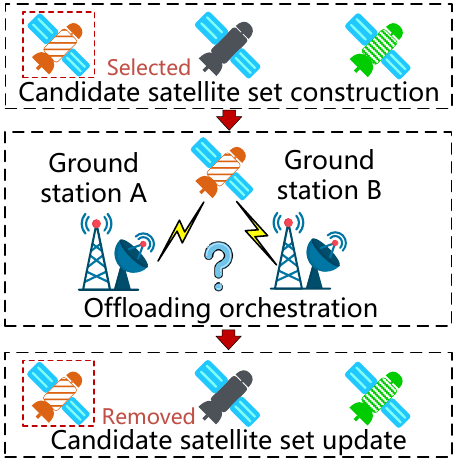}
     \caption{Overview of the \schemeAO{} algorithm.}
     \vspace{-4mm}
     \label{fig:overview_algo}
\end{figure}

\begin{theorem} \label{thm2}
The time complexity of the proposed algorithm is $O((|V_S|(\log|V_S|+\mathcal{RE})$, where $\mathcal{R}=\max_{v\in V_S}|R_v(\tau)|$ and $\mathcal{E}=\max_{\tau\in\mathbb{T}}E_S(\tau)$.
\end{theorem}
\begin{proof}
The initialization in lines \ref{algline:stage1_begin}–\ref{algline:stage1_end} incurs a time complexity of $O(|\overline{V}_S(\tau)|\log |\overline{V}_S(\tau)|)=O(|V_S|\log|V_S|)$ iterations, where the equality holds since $ \overline{V}_S(\tau)\subseteq V_S$. For each satellite  $v\in$  $\overline{V}_S(\tau)$, we determine the target ground station for task offloading, and this process for a single satellite $v$ has a time complexity of $O(|R_v(\tau)|\cdot|E_S(\tau)|)$. The orchestration stage in lines \ref{algline:stage2_begin}–\ref{algline:stage2_end} thus has an overall time complexity of $O(|\overline{V}_S(\tau)|\cdot|R_v(\tau)|\cdot|E_S(\tau)|)=O(|V_S|\cdot|R_v(\tau)|\cdot|E_S(\tau)|)$. In summary, the time complexity of the proposed algorithm is $O(|V_S|\log|V_S|)+O(|V_S|\cdot|R_v(\tau)|\cdot|E_S(\tau)|)=O((|V_S|(\log|V_S|+\mathcal{RE})$, where $\mathcal{R}=\max_{v\in V_S}|R_v(\tau)|$ and $\mathcal{E}=\max_{\tau\in\mathbb{T}}E_S(\tau)$. This indicates that the algorithm runs in quasi-linear time relative to the satellite network scale.
\end{proof}

\section{Performance Evaluation}
\label{sec:eval}

\subsection{Evaluation Settings}

\noindent \textbf{Implementation details.} We implement the SQSO framework with the integrated \schemeAO{} algorithm, and evaluate our solution on real-world Starlink LEO constellation configurations: a baseline of 72 evenly spaced orbital planes, 22 satellites per orbit at 550 km altitude, plus scaled-up setups to validate performance for constellation expansion.

\noindent \textbf{Parameter setups.} The computing devices deployed on LEO satellites are assumed to be commercial off-the-shelf hardware, such as NVIDIA Jetson \cite{slater_jetson_2020}. Since we focus on remote sensing data, we adopt the computing power and processing capability of NVIDIA Jetson devices \cite{nvidia_jetson_nano_product_development}, where the power consumption is $10.72$ W and processing capability is $1.43$ GHz. The required processing cycles per bit is $737.5$ cycles/bit \cite{mao_computational_resource_2017}. Electricity prices are obtained from renewable energy datasets \cite{kaggle_load_wind_solar_prices_hourly, kaggle_day_ahead_spot_price_electricity} and where we set the prices in the range of $\$0.04$/kWh to $\$0.2$/kWh depending on the regions. The data generation rate is 300 Mb/s \cite{wang_earth_observe_2022}, where each Mbps transmission rate consumes $0.08$ W \cite{yang_towards_2016}. The ground station distribution is based on AWS Ground Station locations \cite{aws-ground-station-antenna-locations} and the population distribution \cite{yang_towards_2016}.

\noindent \textbf{Compared schemes}. Two state-of-the-art schemes are compared with our \schemeAO{}, each representing a different design focus for data processing and delivery in satellite networks:
\begin{itemize}[leftmargin=*, itemindent=0pt]
    \item The \schemeCCT{} \cite{chen2022delay} scheme focuses on in-space routing, where data are neither pre-processed nor offloaded to ground stations. Instead, data are routed through ISLs to the LEO satellite located above the destination.
    \item The \schemeHROA{} \cite{li_processing_2022} scheme focuses on pre-processing the region-of-interest data on LEO satellites before offloading them to ground stations. This approach reduces the required downlink bandwidth.  
\end{itemize}

\subsection{Evaluation Results}

\subsubsection{Energy Efficiency and Sustainable Lifespan}

\begin{figure}[t]
    \mbox{
        \begin{minipage}[t]{0.49\linewidth}
            \centering
            \includegraphics[width=\linewidth]{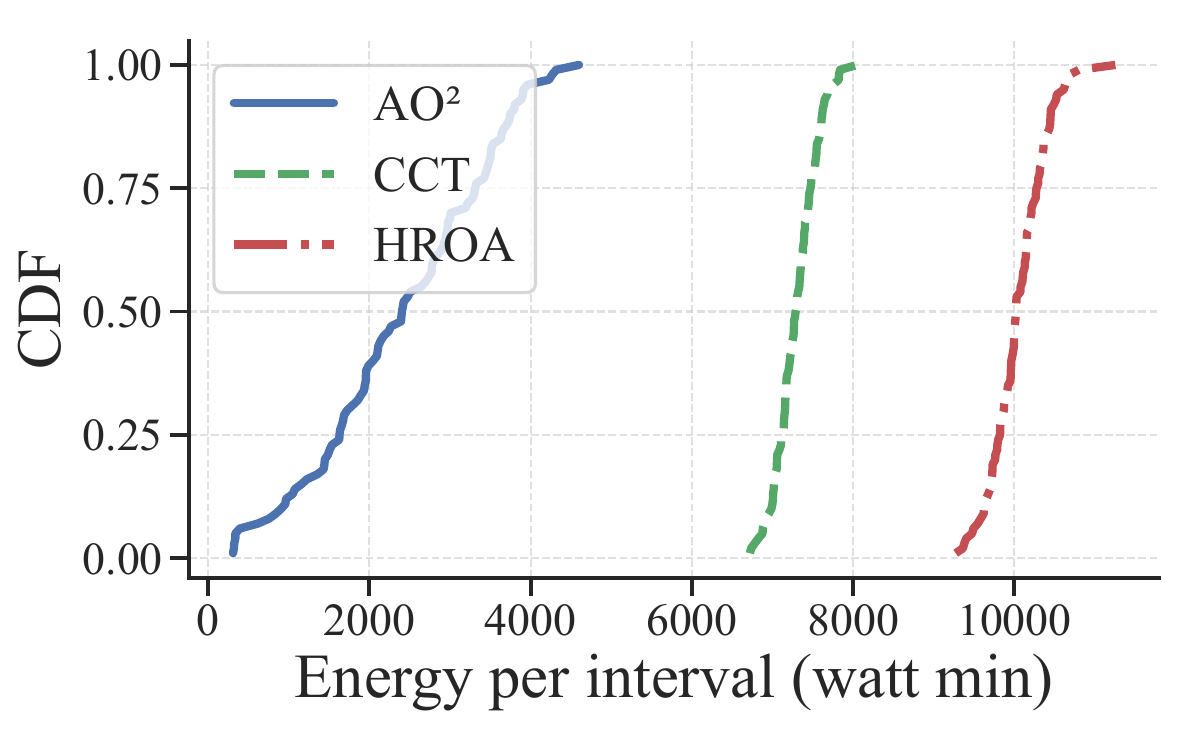} 
            \caption{CDF of total energy consumption per interval.} \label{fig:energy_consumption_cdf}
        \end{minipage}
        \begin{minipage}[t]{0.49\linewidth}
            \centering
            \includegraphics[width=\linewidth]{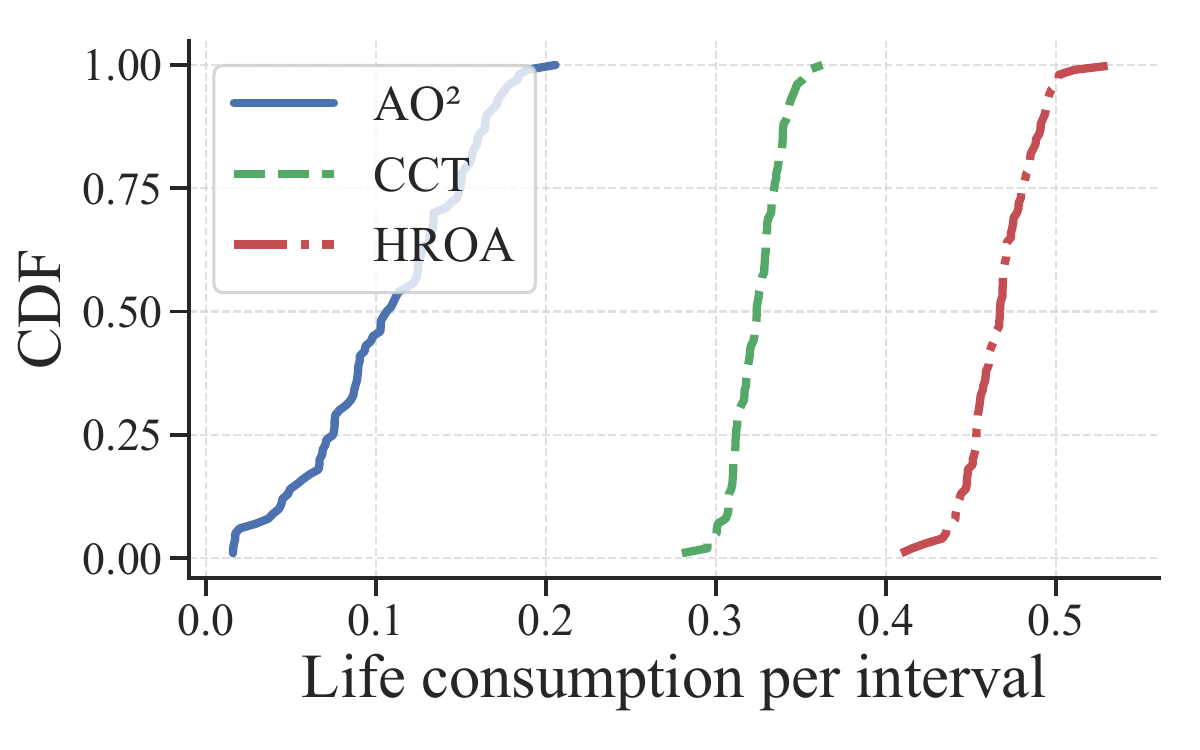} 
            \caption{CDF of total life consumption per interval.} \label{fig:life_consumption_cdf}
        \end{minipage}
    } 
    \vspace{-3mm}
\end{figure}

Fig. \ref{fig:energy_consumption_cdf} shows the CDF of the total energy consumption per interval for the LEO satellite constellation. Our proposed solution achieves \energyImprovementCCT{\%} and \energyImprovementHROA{\%} lower total energy consumption per interval compared with the \schemeCCT{} and \schemeHROA{} schemes, respectively. The \schemeHROA{} scheme consumes significantly more energy because it pre-processes data on LEO satellites before downlink transmission, leading to high CPU and GPU energy usage. Meanwhile, the \schemeCCT{} scheme consumes more energy than our \schemeAO{} because \schemeCCT{} routes data through multiple LEO satellites before reaching the destination, increasing transmission energy consumption.

Fig. \ref{fig:life_consumption_cdf} shows the CDF of total battery life consumption per interval for the LEO satellite constellation. Our solution achieves \lifeImprovementCCT{\%} and \lifeImprovementHROA{\%} lower life consumption compared with the \schemeCCT{} and \schemeHROA{} schemes, respectively. The other schemes have higher life consumption than our solution because higher energy consumption results in deeper battery discharge, leading to higher life consumption.

\subsubsection{Performance Under Different Constellation Sizes}

\begin{figure}[t]
    \mbox{
        \begin{minipage}[t]{0.49\linewidth}
            \centering
            \includegraphics[width=\linewidth]{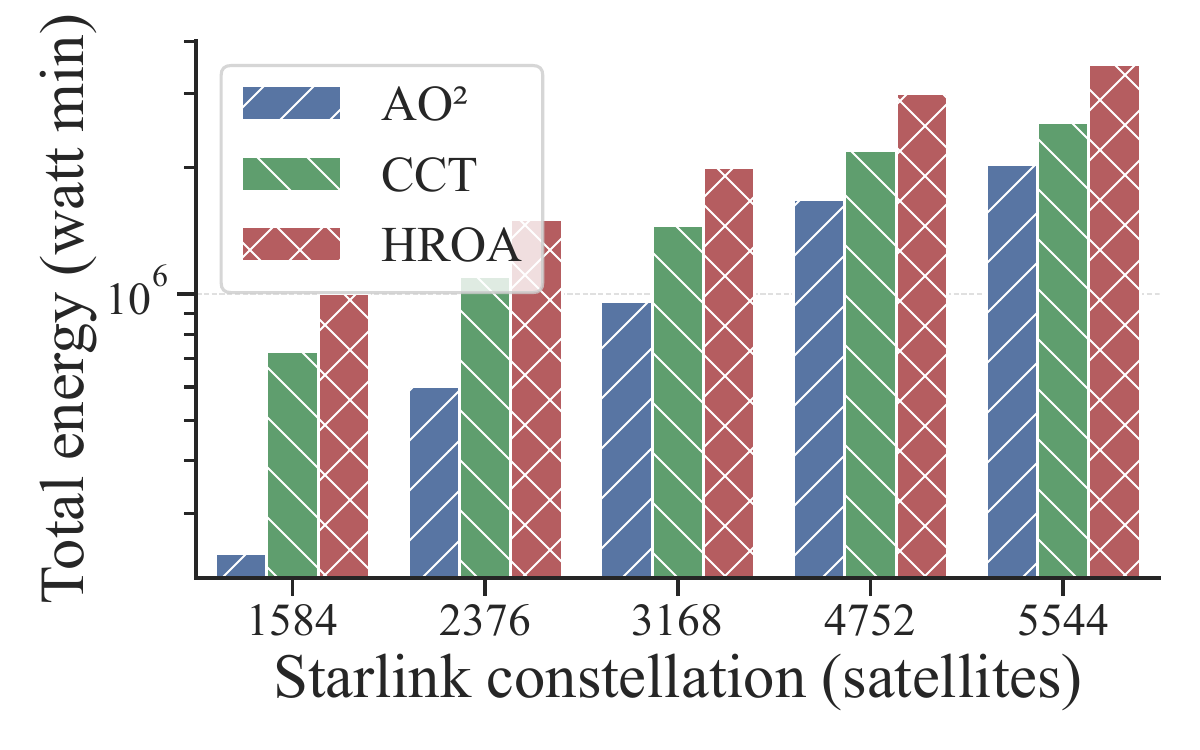} 
            \caption{Total energy consumption in different LEO satellite constellations.} \label{fig:energy_consumption_grouped_bar}
        \end{minipage}
        \begin{minipage}[t]{0.49\linewidth}
            \centering
            \includegraphics[width=\linewidth]{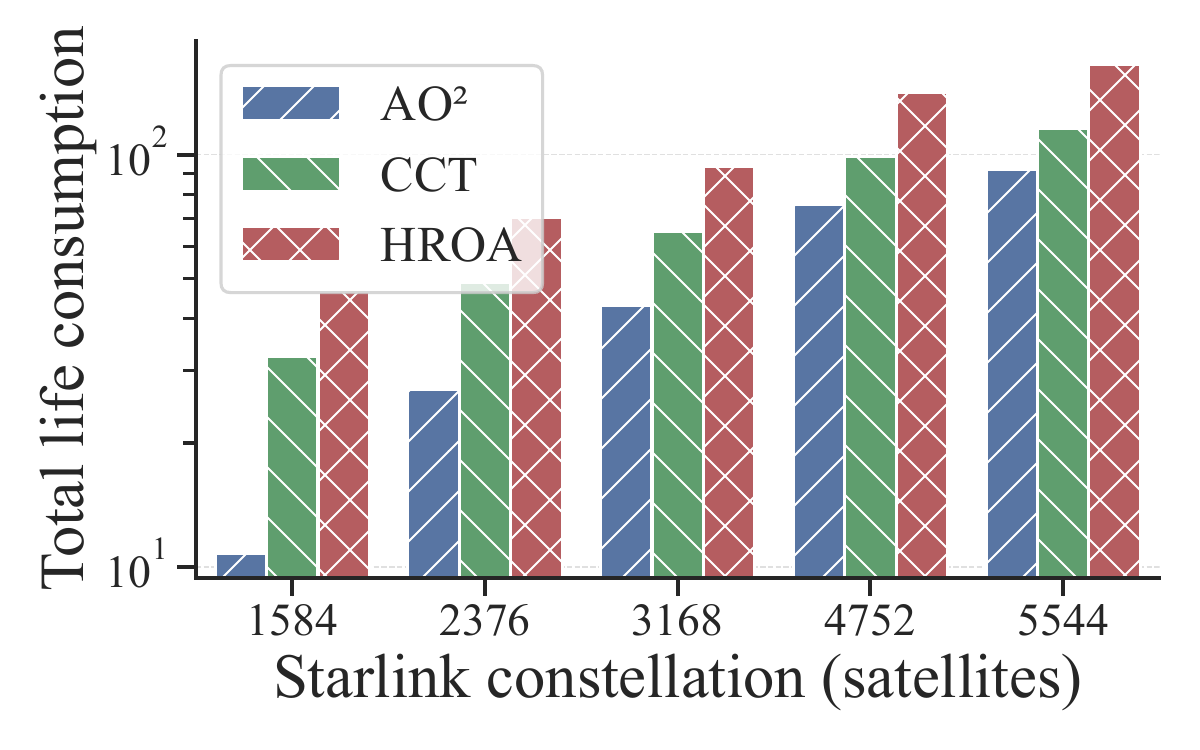} 
            \caption{Total life consumption in different LEO satellite constellations.} \label{fig:life_consumption_grouped_bar}
        \end{minipage}
    } 
\end{figure}

We further evaluate the energy efficiency and sustainability of our solution against other schemes under different Starlink constellation sizes, considering that the Starlink constellation is continuously expanding. 

As illustrated in Fig. \ref{fig:energy_consumption_grouped_bar}, the energy consumption of both our solution and other schemes increases as the constellation size grows. This is due to more data being generated from more LEO satellites. However, our solution \schemeAO{} consistently consumes less energy across different Starlink constellation sizes due to its efficient budget allocation strategy, which prioritizes LEO satellites with higher marginal utility gain and intelligently offloads data to ground stations. 

Due to its consistently low energy consumption, our \schemeAO{} also results in reduced battery life consumption across different constellation sizes, as shown in Fig. \ref{fig:life_consumption_grouped_bar}. This demonstrates that our solution maintains stable and scalable performance as the constellation expands, achieving robust improvements in both energy efficiency and long-term sustainability.

\subsubsection{Performance Under Different Budget Constraints}

\begin{figure}[t]
    \mbox{
        \begin{minipage}[t]{0.49\linewidth}
            \centering
            \includegraphics[width=\linewidth]{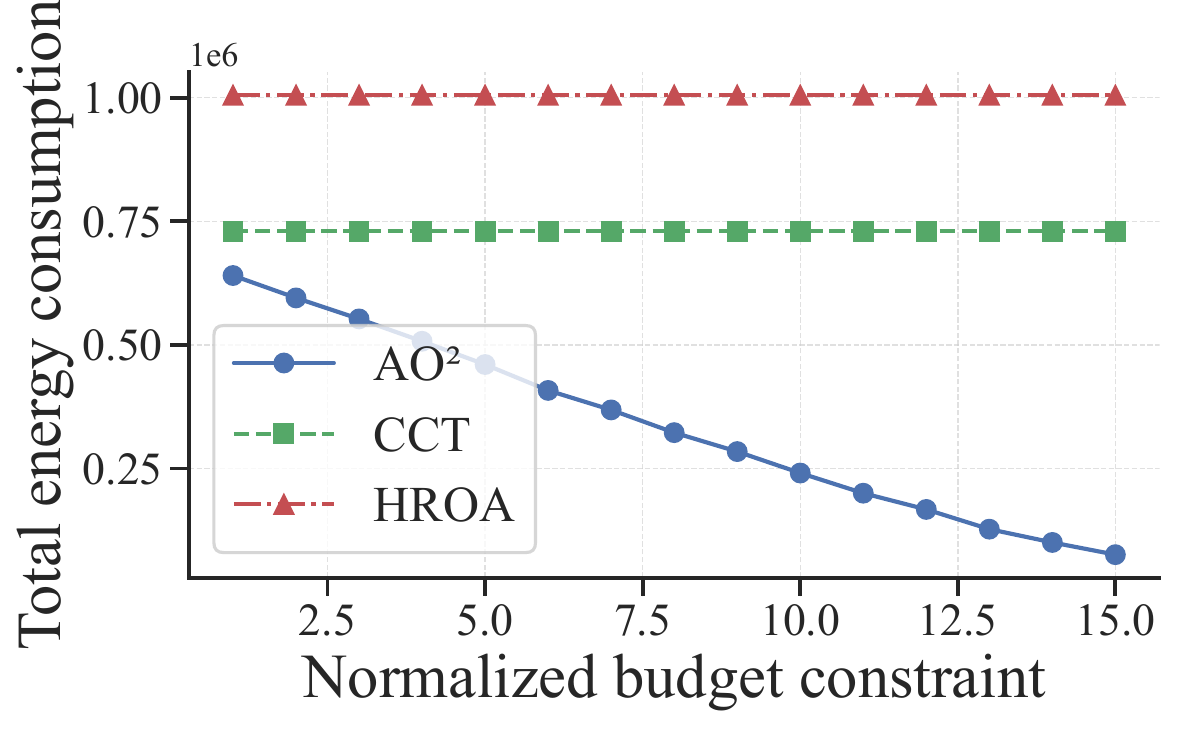} 
            \caption{Energy consumption in different budget constraints.} \label{fig:variable_budget_energy_consumption}
        \end{minipage}
        \begin{minipage}[t]{0.49\linewidth}
            \centering
            \includegraphics[width=\linewidth]{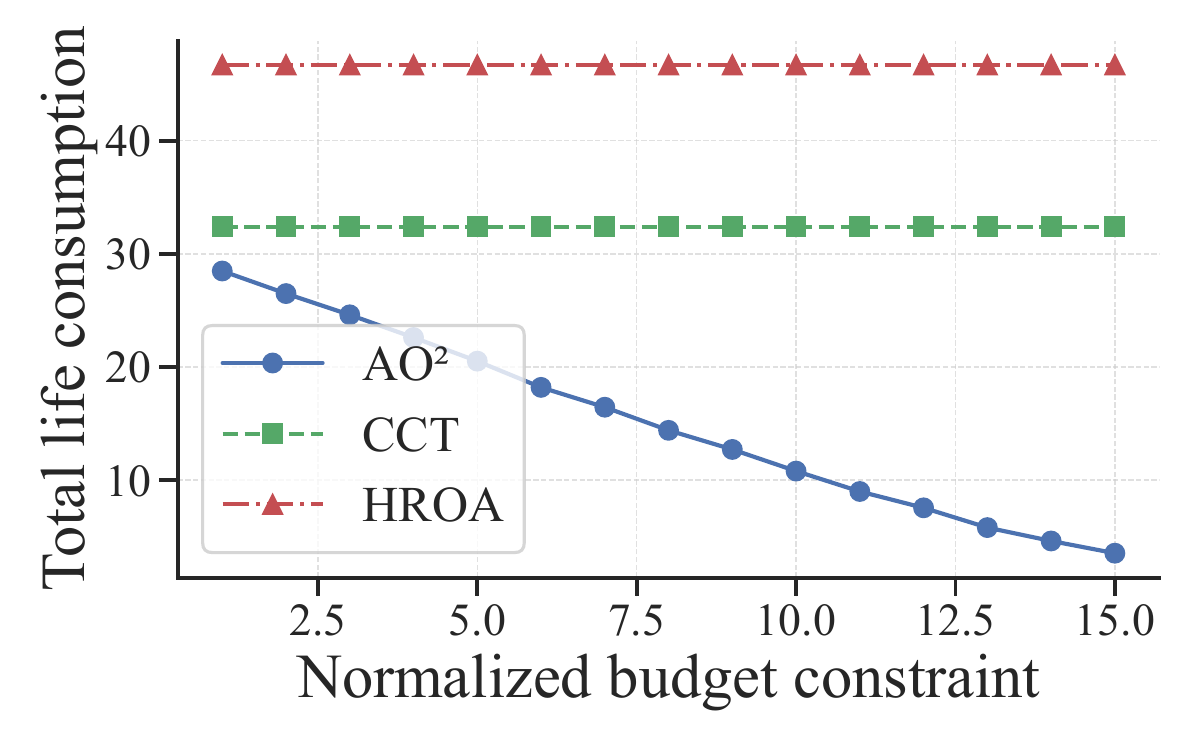} 
            \caption{Life consumption in different budget constraints.} \label{fig:variable_budget_life_consumption}
        \end{minipage}
    } 
    \vspace{-3mm}
\end{figure}

Fig. \ref{fig:variable_budget_energy_consumption} and Fig. \ref{fig:variable_budget_life_consumption} illustrate the energy consumption and battery life consumption under different budget constraints, respectively. The \schemeHROA{} scheme has the highest and nearly constant energy and life consumption across all budget constraints because it focuses on pre-processing data on LEO satellites, which consumes a high amount of energy to reduce data size before transmission, regardless of the budget constraints. For a similar reason, the \schemeCCT{} scheme routes data through multiple LEO satellites via ISLs to reach the destination, leading to consistently high energy and life consumption regardless of the budget constraints as well. 

Our proposed \schemeAO{} effectively utilizes the limited budget to offload data to ground stations in different budget constraints. When budget constraints are tight, the performance gap becomes smaller since insufficient budget provides less flexibility for the algorithm to offload data. However, when the budget becomes moderate or sufficient, our \schemeAO{} can effectively offload more tasks, which further reduces energy and battery life consumption compared to the other schemes.

\subsubsection{Delay Under Different Constellation Sizes}

\begin{figure}[t]
    \mbox{
        \begin{minipage}[t]{0.49\linewidth}
            \centering
            \includegraphics[width=\linewidth]{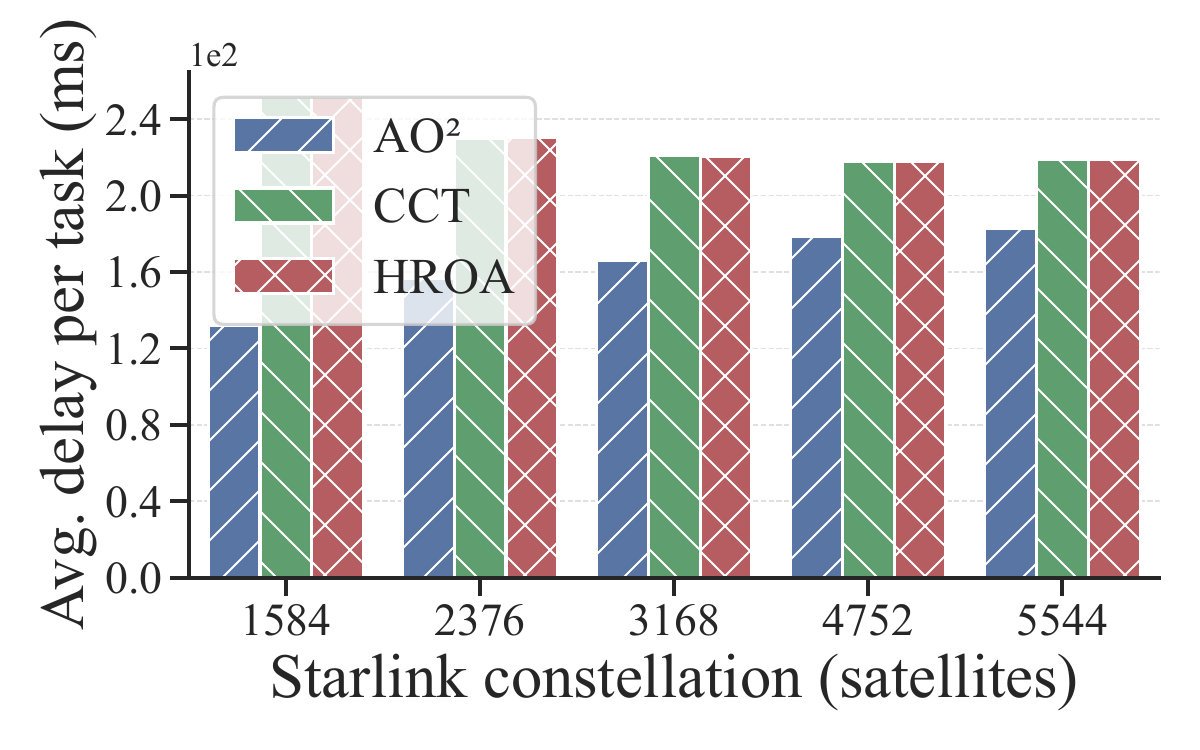} 
            \caption{Average delay grouped bars.}            
            \label{fig:average_delay_grouped_bar}
        \end{minipage}
        \begin{minipage}[t]{0.49\linewidth}
            \centering
            \includegraphics[width=\linewidth]{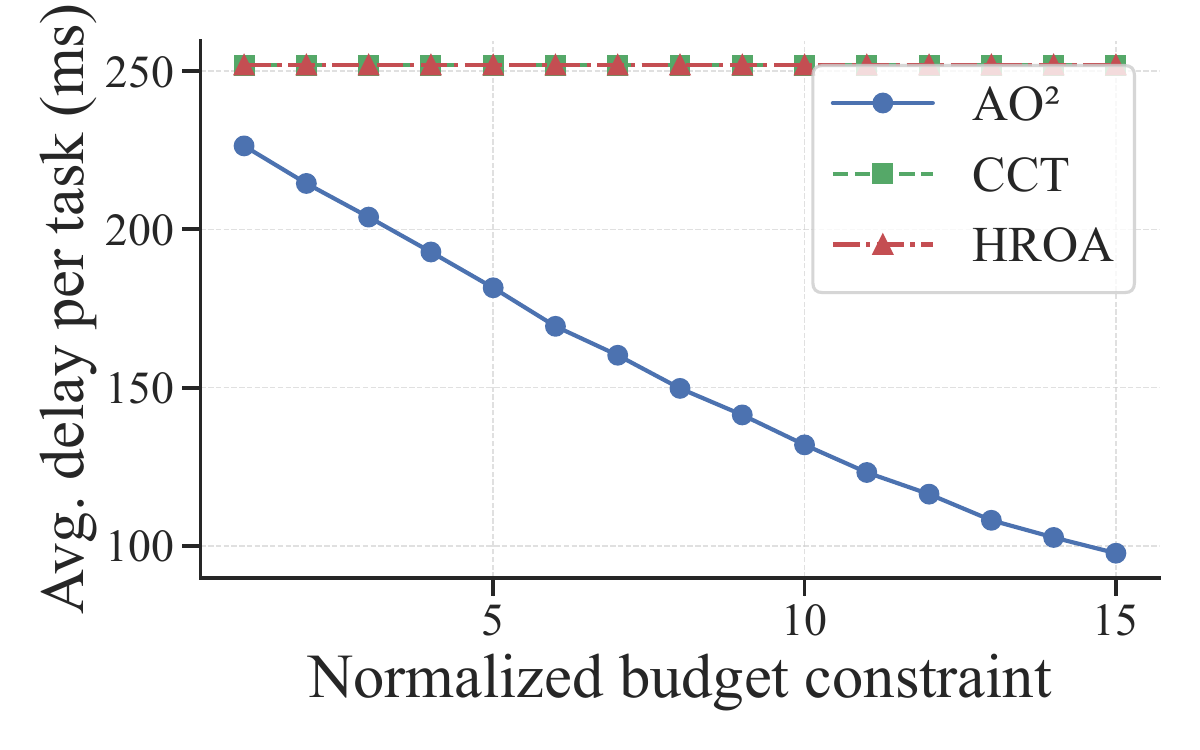} 
            \caption{Average delay in different budget constraints.} 
       \label{fig:avg_delay_in_different_budget_constraints}
        \end{minipage}
    } 
\end{figure}

\begin{table}[t]
\centering
\resizebox{\columnwidth}{!}{
\begin{tabular}{|c|c|c|c|c|c|}
\hline
\diagbox[width=3.6cm, height=1cm]{Runtime}{Starlink satellites} & 1584 & 2376 & 3168 & 4752 & 5544 \\
\hline
Avg. runtime per interval (s) & 10.87 & 14.95 & 17.87 & 26.09 & 30.66 \\
\hline
Std. runtime per interval (s) & 1.62 & 2.73 & 1.55 & 1.12 & 1.24 \\
\hline
Avg. runtime per task (ms) & 10.75 & 9.91 & 8.87 & 8.65 & 8.68 \\
\hline
Std. runtime per task (ms) & 1.64 & 1.85 & 0.73 & 0.34 & 0.30 \\
\hline
\end{tabular}
}
\caption{Algorithm runtime across network sizes.}
\label{tab:our_algorithm_runtime_network_sizes}
\vspace{-6mm}
\end{table}

Fig. \ref{fig:average_delay_grouped_bar} shows the average delay per task under different Starlink constellation sizes. Our solution \schemeAO{} consistently achieves the lowest average delay per task compared to the other schemes under different Starlink constellation sizes. This is primarily because data are offloaded to ground stations and subsequently routed through terrestrial infrastructure, where network resources are more abundant and routing paths are more stable than in LEO satellite networks.

Fig. \ref{fig:avg_delay_in_different_budget_constraints} shows the average delay in different budget constraints, where our solution has the lowest average delay per task across budget constraints. This is because our solution incorporates QoS utility into the marginal utility gain formulation, which efficiently uses the limited budget to offload tasks. This suggests that our solution maintains robust QoS performance as the constellation scales and budget constraint changes, highlighting its scalability and effectiveness in evolving LEO networks and changing budget constraints.

\subsubsection{Algorithm Runtime Performance}
TABLE \ref{tab:our_algorithm_runtime_network_sizes} reports the average per-interval and per-task runtime of our algorithm under varying Starlink constellation sizes. The per-interval runtime grows approximately linearly with constellation scale, which aligns perfectly with the theoretical conclusion in Theorem \ref{thm2}, driven by the increased task volume from more LEO satellites in larger constellations. This linear growth characteristic demonstrates the strong scalability of our proposed algorithm for large constellations.

\noindent \textbf{Runtime per interval optimization.} The current implementation of our proposed algorithm computes the average algorithm runtime per interval using a sequential task-processing strategy. Due to the inherent independence among tasks within each interval, the algorithm can be further optimized through parallelization. Specifically, the total budget can be partitioned into $n$ disjoint subsets to create $n$ pipelines. At the end, the remaining budget from each subset can be aggregated to form a new budget set to process the remaining tasks. This parallel approach can reduce the average algorithm runtime by approximately a factor of $n$.

Moreover, the average algorithm runtime per task decreases as the constellation size grows. In smaller constellations, the average algorithm runtime per task appears larger because the runtime is at the millisecond scale, where small fluctuations have a stronger visual impact. As more LEO satellites are deployed in larger constellations, the variation of the algorithm runtime per task becomes smaller, approaching a more stable average value. We expect the average algorithm runtime per task to converge as the constellation size continues to grow.
\section{Related Work}\label{sec:related_work}

Our work relates to two main directions: (i) in-orbit processing that pushes computation and service functionality into satellites and constellations, and (ii) satellite offloading and geo-distributed scheduling under dynamic constraints.

\noindent \textbf{In-orbit processing.}  
Recent systems research explores transforming satellites into computing platforms. Empirical studies characterize the feasibility and constraints of running commodity workloads in space \cite{xing2024cots, li_processing_2022}. Other efforts redesign network functionality for orbital environments, e.g., decoupling control-plane state to tolerate mobility and intermittent connectivity \cite{li2022spacecore}, or introducing constellation-scale in-orbit caching for CDN-like delivery \cite{zheng2025starcdn}. Application-level systems such as Serval further demonstrate the benefits of jointly leveraging satellite and ground compute for Earth-observation analytics \cite{tao2024serval}.

These works primarily aim to enhance orbit-native functionality. In contrast, we do not push more services into space; instead, we study satellite-grid co-design, opportunistically shifting workloads to renewable-adjacent ground compute under contact and energy constraints.

\noindent \textbf{Offloading and scheduling.}  
Prior work studies data transfer and task scheduling over intermittent satellite contacts. L2D2 and Umbra model orbital dynamics and heterogeneous bottlenecks to reduce LEO data latency through scheduling and rate adaptation \cite{vasisht2021l2d2, vasisht2023umbra}. Satellite edge computing research further considers delay-aware task offloading under satellite energy and computation limits \cite{chen2022delay, chou2025commercial}.

% Beyond satellite systems, geo-distributed computing literature optimizes placement and routing across datacenters under latency and bandwidth cost trade-offs \cite{wendell2010donar, uluyol2020nsdi, jain2023skyplane}.

Beyond satellite systems, adjacent aerial and spacecraft systems have explored service coordination and autonomous decision-making under dynamic topologies and limited onboard resources \cite{11300881, 11169619, 9767753}. 

Different from prior work, we jointly couple (i) contact-window feasibility, (ii) time-varying electricity budgets driven by renewable curtailment and price dynamics, and (iii) application QoS into a single online orchestration problem, explicitly linking offloading decisions to both task performance and satellite battery sustainability.

\section{Conclusion}\label{sec:conclusion}
\label{conclusion} 
We present SQSO, a renewable-aware task offloading framework for satellite-integrated smart grid systems, paired with a corresponding adaptive offloading orchestration algorithm. By jointly optimizing QoS, battery degradation, contact-window feasibility, and electricity budgets, SQSO bridges satellite networking and terrestrial energy systems via a unified orchestration layer. Our algorithm schedules offloading decisions based on marginal utility gain, consistently reducing on-board energy consumption, battery degradation, and communication delay across diverse constellation sizes and budget constraints.

\section*{Acknowledgements}
\label{acknowledgements}
This research was generously supported in part by the Carbon Neutrality and Energy System Transformation (CNEST) Program, Tsinghua University, and the National Key R\&D Program Project (grant number 2024YFE0203700).

\bibliographystyle{IEEEtran} % We choose the "plain" reference style
\bibliography{sections/reference} % Entries are in the refs.bib file

\end{document}